\documentclass{article}

\usepackage[english]{babel}

\usepackage[letterpaper,top=2cm,bottom=2cm,left=3cm,right=3cm,marginparwidth=1.75cm]{geometry}

\usepackage{amsmath}
\usepackage{amssymb}
\usepackage{amsthm}
\usepackage{graphicx}
\usepackage{mathtools}
\usepackage[colorlinks=true, allcolors=blue]{hyperref}
\usepackage{authblk}

\usepackage{tikz}
\usepackage{enumitem} \setlist[itemize]{noitemsep}

\newtheorem{theorem}{Theorem}
\newtheorem*{remark}{Remark}
\newtheorem{proposition}[theorem]{Proposition}
\newtheorem{definition}[theorem]{Definition}

\newcommand{\TMD}{\mathrm{TMD}}
\newcommand{\TMDg}{\mathrm{TMDg}}

\newcommand{\diag}{\mathrm{diag}}

\renewcommand{\epsilon}{\varepsilon}

\title{Stability of topological descriptors for neuronal morphology}

\author[1]{David Beers\thanks{Email: beers@maths.ox.ac.uk}}
\author[1,2]{Heather A. Harrington}
\author[1]{Alain Goriely}
\affil[1]{Mathematical Institute, University of Oxford}
\affil[2]{Wellcome Centre for Human Genetics,  University of Oxford}
\date{}

\begin{document}
\maketitle

\begin{abstract}
The topological morphology descriptor of a neuron is a multiset of intervals associated to the shape of the neuron represented as a tree. In practice, topological morphology descriptors are vectorized using persistence images, which can help classify and characterize the morphology of broad groups of neurons. We study the stability of topological morphology descriptors under small changes to neuronal morphology.  We show that the persistence diagram arising from the topological morphology descriptor of a neuron is stable for the 1-Wasserstein distance against a range of perturbations to the tree. These results guarantee that persistence images of topological morphology descriptors are stable against the same set of perturbations and reliable.
\end{abstract}

\section{Introduction}
Topological data analysis is a field in applied mathematics concerned with the shape of data. Arising from this field, the celebrated persistent homology (PH) algorithm  \cite{zomorodian2005computing} records relevant information regarding geometric data. In certain settings, PH sidesteps alignment issues of other geometric descriptors. The persistent homology algorithm takes as input a topological space $X$ equipped with a function $f$, and returns a multiset of intervals called a barcode, or equivalently, a multiset of points in the plane called a persistence diagram. Roughly speaking, the $x$ and $y$ coordinates of each point in a persistence diagram detail when a particular topological feature appears and then subsequently disappears in the sequence of shapes $f(-\infty,t]$ as $t$ varies. The space of persistence diagrams can be equipped with the bottleneck and 1-Wasserstein distances, defined in Section \ref{sec:metrics}. Loosely speaking, the bottleneck distance between two persistence diagrams is the smallest number $\epsilon$ that exists allowing a pairing between points in each diagram moving points no more than $\epsilon$, whereas for the 1-Wasserstein distance $\epsilon$ controls the cumulative distance the points are moved in a matching. Crucially, persistence diagrams arising from PH are stable against certain perturbations of $f$ with respect to these metrics \cite{bubenik2014categorification, cohen2005stability, skraba2020wasserstein}. Persistence diagrams can be vectorized for machine learning \cite{adams2017persistence,bubenik2015statistical}. One such vectorization, persistence images, is stable under perturbations to the 1-Wasserstein distance \cite{adams2017persistence}.

Ideas from topological data analysis have been developed to study neurons. Neuronal shape plays a key role in cognition \cite{london2005dendritic}; therefore, the morphology of neurons is often analysed and classified. Mathematicians have proposed  specialized topological descriptors for neuronal morphology. Li et al. \cite{li2017metrics} showed that PH can be directly applied to neuronal data to encode morphologically relevant information in persistence diagrams. Kanari et al. \cite{kanari2018topological} proposed the topological morphology descriptor (TMD), a barcode derived from neuronal data using a PH-inspired algorithm. This TMD barcode is derived from a topological space $X$ equipped with a function $f$, but here $X$ is required to be a rooted tree $T$. The present authors proved in \cite{beers2022barcodes} that when $f$ increases along paths away from the root, this methodology coincides with that of \cite{li2017metrics}. Kanari et al. \cite{kanari2018topological} showed that the TMD is stable against perturbations of neuronal data for the bottleneck distance. The TMDs of different neurons can be vectorized using persistence images, which can be averaged, interpreted and compared. 
Khalil et al. \cite{khalil2022topological} state the TMD is stable for the 1-Wasserstein distance against perturbations to an
embedding of $T$ in Euclidean space when $f$ the Euclidean distance of the
vertices to the soma; no quantification or bounds are given measuring the stability of these perturbations.

\begin{figure}[htbp]
\centering
\resizebox{.8\textwidth}{!}{
\includegraphics{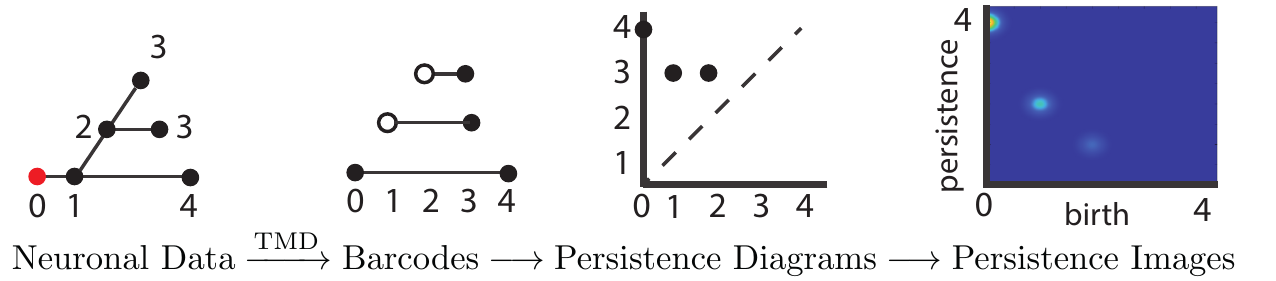}
}
\caption{The persistence image pipeline applied to neuronal data.}
\label{fig:PIpipeline}
\end{figure}

Here, we prove 1-Wasserstein stability of the TMD against perturbations to any function $f$ and to a class of perturbations to the structure of $T$. We provide Lipschitz bounds for both kinds of stability. In other words, we show that the composite of the first two steps in the Figure \ref{fig:PIpipeline} pipeline of converting neuronal data to persistence images is stable in greater generality. As a direct consequence of the work of Adams et al. \cite{adams2017persistence}, the entire pipeline of Figure \ref{fig:PIpipeline} is stable against the perturbations we consider.
The key to our main stability result is the identification of classes of perturbation that are stable for the 1-Wasserstein distance. Roughly speaking, these perturbations correspond to the following perturbations to a tree:
\begin{itemize}
    \item moving slightly the positions of the neuron's branches;
    \item adding a short branch;
    \item deleting a short branch.
\end{itemize}
The rigorous definition of these perturbations will be given in Section \ref{sec:results}. The main result of this paper, Theorem \ref{thm:main}, is that the TMD is stable against these perturbations for the 1-Wasserstein distance. Therefore, persistence images are stable against the same perturbations.

\section{Background}
\subsection{The topological morphology descriptor (TMD)}
Let $T$ denote a finite rooted tree with root $r$. In this paper we say a non-root vertex $l\in N(T)$ is a \textit{leaf} if it is a vertex of degree one. The root is only called a leaf if it has degree zero. A vertex in $T$ is called a \textit{branch point} if it has degree three or greater. The root, again an exception, we call a branch point if it has degree two or greater. We denote by $N(T)$, $E(T)$, and $L(T)$ the sets of vertices, edges, and leaves of $T$ respectively. We define the \textit{depth} of a vertex $v\in N(T)$ to be the number of nodes, minus one, along the shortest path from $r$ to $v$. The \textit{depth} of a rooted tree $T$ is the greatest depth of any of its vertices. Consider a real valued function $f$ on $N(T)$. The \textit{topological morphology descriptor} (TMD) of $T$ and $f$, denoted $\TMD(T,f)$, is a multiset of intervals called a barcode obtained via the algorithm developed in \cite{kanari2018topological}. If $T$ is a tree representing a neuron with its root representing the soma and $f$ records a notion of distance of the vertices of $T$ to the root, then $\TMD(T,f)$ is interpreted as a decomposition of the underlying neuron into its branches, with each interval recording how far a branch's initial and terminal points are from the root.

For any vertex $v$ of $T$ with child $v'$ we call the rooted tree generated by $v'$, its children, its children's children, and so on, along with the edge $(v,v')$, a \textit{child branch} of $v$. The process of obtaining a barcode from $T$ and $f$ is as follows:
\begin{enumerate}
    \item Choose any branch point $b$ in $T$;
    \item Identify one of the children $c$ of $b$ whose child branch maximizes $f$ on the vertices that are leaves of $T$ descendant from $b$;
    \item Detach all child branches for children $c'\neq c$ of $b$. Replace $T$ with the resulting forest;
    \item If there are no branch points in the resulting forest, then $T$ is a collection of intervals, and we are done. Otherwise return to step 1.
\end{enumerate}
What will remain at the end of this procedure is a multiset of intervals with right endpoints in bijective correspondence with the leaves of $T$. We assign numbers to the endpoints of these intervals via $f$ and call the resulting structure $\TMD(T,f)$. We refer to any collection of intervals with endpoints labelled by real numbers as a \textit{barcode}. Figure~\ref{fig:TMDalg} shows an example of a barcode computed via the TMD. An efficient algorithm for computing the barcode $\TMD(T,f)$ is given in \cite{kanari2018topological}. Note that in general it may be the case that the number assigned by $f$ to the left endpoint of an interval may be greater than that assigned to the right endpoint of an interval. However, this never happens if $f$ is increasing along paths away from the root, as a consequence of \cite[Theorem 2]{beers2022barcodes}. We then define $\TMDg(T,f)$ to be the disjoint union of a copy $(a,b)\in\mathbb{R}^2$ for every interval in $\TMD(T,f)$ with left endpoint $a$ and right endpoint $b$, along with infinite copies of $(x,x)\in\mathbb{R}^2$ for every real number $x$\footnote{In \cite{kanari2018topological}, the authors use reversed notation, where an interval with left endpoint $a$ and right endpoint $b$ gets mapped to $(b,a)$. This is likely done to suggest a connection between the TMD and extended persistence of the superlevel set filtration of $f$.}. By convention, we refer to the value $a$ the \textit{birth} of a feature represented by a point $(a,b)$, and we refer to $b$ and $b-a$ as the \textit{death} and \textit{persistence} of the same feature, respectively. Notice that $\TMDg(T,f)$ will always contain a point $(f(r), L)$, where $L$ is the maximum of $f$ on the leaves of $T$. For convenience, we denote the multiset of all diagonal points $(x,x)$, each with infinite multiplicity, as $\diag$.

\begin{figure}[htbp]
\centering
\resizebox{\textwidth}{!}{
\includegraphics{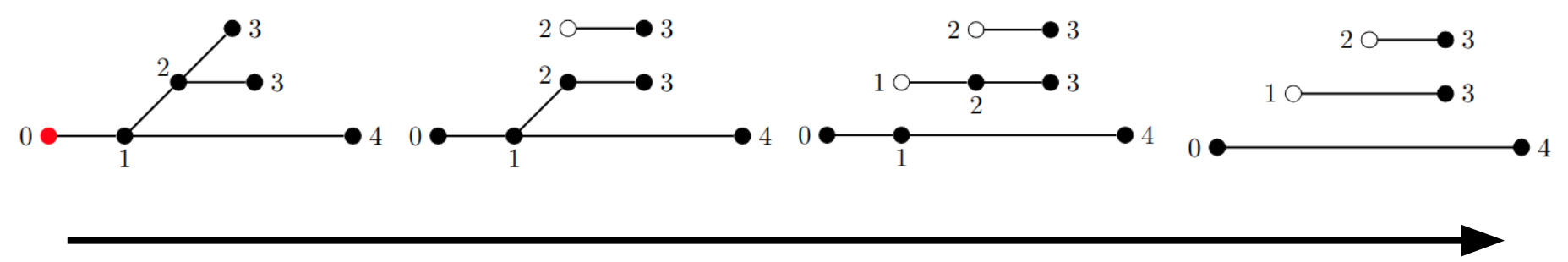}
}
\caption{Computing the TMD from neuronal data.}
\label{fig:TMDalg}
\end{figure}

\subsection{Metrics on the space of persistence diagrams}
\label{sec:metrics}

The multiset $\TMDg(T,f)$ is an example of what is called a \textit{persistence diagram}, a multiset of $\mathbb{R}^2$ containing $\diag$ and only finitely many off-diagonal points. There are many metrics on the space of persistence diagrams. In this paper we consider two of particular interest.

\begin{definition}
A map $\phi$ between persistence diagrams $D_1$ and $D_2$ is called a matching if it is bijective. The bottleneck distance between $D_1$ and $D_2$ is defined to be
\begin{equation*}
    d_B(D_1,D_2) := \inf_\phi \sup_{x \in D_1} \|x - \phi(x)\|_\infty,
\end{equation*}
where the infimum is taken over all matchings $\phi$ between $D_1$ and $D_2$. The 1-Wasserstein distance is defined similarly, as
\begin{equation*}
    d_1(D_1,D_2) := \inf_\phi \sum_{x\in D_1} \|x-\phi(x)\|_\infty,
\end{equation*}
again with the infimum being over matchings $\phi$ between $D_1$ and $D_2$.

If for a matching $\phi$ we have
\begin{equation*}
    \sum_{x\in D_1} \|x-\phi(x)\|_\infty \leq \epsilon,
\end{equation*}
we will say that $\phi$ is an $\epsilon$-matching\footnote{It is more common in the literature for an $\epsilon$-matching to refer to a matching such that $\sup_{x\in D_1} \|x-\phi(x)\|_\infty \leq \epsilon$. We use an alternate definition in this work as the 1-Wasserstein distance is of primary interest here.}.
\end{definition}
Edelsbrunner and Harer showed \cite[Section 8]{edelsbrunner2022computational} that $d_B$ and $d_1$ are indeed metrics on the space of persistence diagrams. By mapping $\TMD(T,f)$ to $\TMDg(T,f)$, we attain a pseudometric on the space of TMDs from either of these metrics on the space of persistence diagrams. Explicitly, we define
\begin{equation*}
\begin{split}
    d_B(\TMD(T,f), \TMD(T',f')) &: = d_B(\TMDg(T,f), \TMDg(T',f')),\\
    d_1(\TMD(T,f), \TMD(T',f')) &: = d_1(\TMDg(T,f), \TMDg(T',f')).
\end{split}
\end{equation*}

Stability results for the TMD algorithm with respect to the bottleneck distance were given by Kanari et al. \cite[SI, Section 4]{kanari2018topological}.
Khalil et al. \cite{khalil2022topological} state the TMD is stable against perturbations to an embedding of $T$ in Euclidean space for the 1-Wasserstein distance when $f$ is the Euclidean distance of the vertices to the soma; however Lipschitz bounds for this kind of stability are not provided. We prove in Section~\ref{sec:results} that the TMD is stable against perturbations of arbitrary functional information $f$, and against certain perturbations of the underlying rooted tree $T$. Before we prove these results we motivate the importance of 1-Wasserstein stability in the next subsection.

\subsection{Persistence images}
For data science problems, persistence diagrams are often mapped into vector spaces, which allows them to be analyzed via machine learning \cite{bubenik2015statistical, carriere2015stable, robins2016principal}. One popular technique of this style is to transform persistence diagrams into matrices called \textit{persistence images} introduced by Adams et al \cite{adams2017persistence}. Persistence images are of particular interest in the setting of the TMD as they are the standard type of vectorization used in this context \cite{kanari2018topological,kanari2019objective,laturnus2020systematic} in part because they provide interpretable visual summaries of averaged persistence diagrams.

Obtaining a persistence image from a persistence diagram is a process involving three steps. The first step is to transform a persistence diagram $D$ to another multiset of the plane $\tilde{D}$ by mapping each point $(x,y)$ to $(x,y-x)$. Next we fix a positive number $\sigma$ and define a function $f_\sigma$ to be the sum of 2D Gaussian functions of standard deviation $\sigma$ centered over each point in $\tilde{D}$, and weighted by their distance from the $x$-axis. The function $f_\sigma$ is called a \textit{persistence surface}. A persistence image can then be defined by fixing a grid over a relevant range of values and generating a matrix where the $(i,j)^\mathrm{th}$ value is the integral of $f_\sigma$ over the $(i,j)^\mathrm{th}$ cell. If the dimensions of each cell in the grid are small relative to $\sigma$, a persistence image can be approximated by sampling over a grid of points. Persistence images can be viewed as vectors with a coordinate for each cell in their corresponding grid. Viewing persistence images as vectors, they are stable in their 1, 2, and $\infty$ norms with respect to the 1-Wasserstein distance of persistence diagrams \cite[Theorem 10]{adams2017persistence}\footnote{The observant reader may notice that the Lipschitz bound provided in the cited theorem is dependent on the $\infty$-norm of our weighting function. This may cause some alarm since we suggest weighting by distance to the $x$-axis, an unbounded function. However, going through the proof one observes that the bound is only dependent on the $\infty$-norm of the weight function over points in a fixed dataset, which is always bounded, provided the dataset is finite and each diagram has finitely many nontrivial points. Furthermore, the weighting the authors themselves use in \cite[page 8]{adams2017persistence} differs from the distance from the $x$-axis function by a constant multiple on data.}.

The first transformation step $(x,y) \mapsto (x, y-x)$ can be removed from the construction of persistence images without disrupting stability guarantees, provided we instead apply a weight of $|y-x|$ to Gaussians centered at $(x,y)$ when constructing persistence surfaces. Indeed, following the proofs of Theorems 9 and 10 from \cite{adams2017persistence}, it can be observed that the same reasoning holds. In fact, the bounds we are guaranteed are tighter by a factor of $\sqrt{5}$, thanks to no longer needing to transform coordinates in the proof of Theorem 9. This can be a slightly convenient alternate version of persistence images in practice, since in certain settings $\TMDg(T,f)$ can be a persistence diagram with points in both regions $y<x$ and $y>x$. Similarly, we can also reflect points in a persistence diagram without the transformation $(x,y)\mapsto(x,y-x)$ about the diagonal $y=x$, or with the transformation about the $x$-axis while still ensuring persistence images are 1-Wasserstein stable. We mention this since Kanari et al. have already used unweighted analogues of reflected persistence images without the transformation $(x,y)\mapsto(x,y-x)$ in their previous work on the TMD \cite{kanari2018topological,kanari2019objective}. In summary, any transformation of a tree equipped with a function that is stable for the 1-Wassertein distance of its TMD is also stable with respect to popular metrics on its associated transformed and untransformed, reflected and unreflected persistence images. This fact motivates the following section.

\section{1-Wasserstein Stability}
\label{sec:results}
In this section we prove the TMD is stable to four types of perturbations. We build on an inductive argument of \cite[SI]{kanari2018topological}. Denote \textbf{P1}, \textbf{P2}, \textbf{P3}, and \textbf{P4} the four kinds of perturbations to data from $(T,f)$ to $(T',f')$:

\begin{itemize}
    \item [\textbf{P1}] \textit{Vertex perturbation.} $T' = T$ and $f'$ satisfies $|f'(v)-f(v)|\leq \epsilon$ for all vertices $v$ in $N(T)$.
    \item [\textbf{P2}] \textit{Attaching an edge to a vertex.} $N(T') = N(T) \cup \{v'\}$, $E(T') = E(T) \cup \{(v,v')\}$ for some vertex $v\in N(T)$, $f' = f$ on $N(T)$, $|f'(v') - f(v)|\leq \epsilon$, and $f(v)\leq f(l) + \epsilon$ for at least one leaf $l$ of $T$ that is a descendant of $v$. Put another way, we attach a short edge to one of the vertices of $T$, one of whose descendant leaves does not have a much smaller $f$ value than its own.
    \item [\textbf{P3}] \textit{Attaching an edge to an edge.} There is an edge $(n,n')$ in $E(T)$, $N(T') = N(T) \cup \{v,v'\}$, $E(T') = E(T) \cup \{(n,v),(v,n'),(v,v')\} - \{(n,n')\}$, $f=f'$ on $N(T)$, $|f(v) - f(v')|\leq \epsilon$, and $f'(v)\leq f(l) + \epsilon$ for at least one leaf $l$ of $T$ that is a descendant of $v$. In words, we add a vertex interior to an edge of $T$, ensure that one of the descendant leaves of this vertex does not have a much smaller $f$ value than its own, and attach a short edge to this vertex.
    \item [\textbf{P4}] \textit{Edge retraction.} $(T,f)$ is obtained by a perturbation of type \textbf{P2} or \textbf{P3} from $(T',f')$.
\end{itemize}
We leave the root vertex unchanged in all of these classes of perturbation.

We show visualizations of these transformations in Figure \ref{fig:trans}, along with a type of transformation we do not claim to be stable.

\begin{figure}[htbp]
\centering
\resizebox{\textwidth}{!}{
\includegraphics{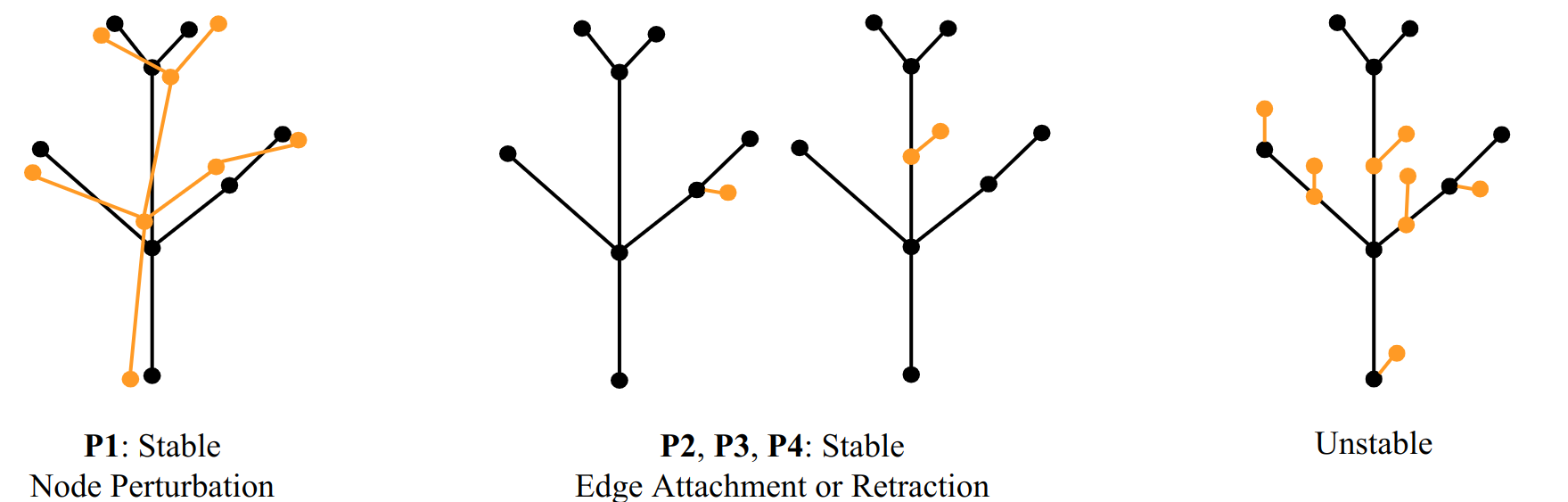}
}

\caption{Tree perturbations we show are stable, and one type of perturbation that is unstable.}
\label{fig:trans}
\end{figure}

Now we state the main theorem of the paper.

\begin{theorem}
Let $T$ be a rooted tree and $f:T \to \mathbb{R}$. Let also $T'$ be a rooted tree with $f':T' \to \mathbb{R}$.
\begin{enumerate}
    \item If $(T',f')$ is obtained by a perturbation of type \textbf{P1} of $(T,f)$, then
    \begin{equation*}
        d_1(\TMD(T,f),\TMD(T',f')) \leq |L(T)|\epsilon.
    \end{equation*}
    
    \item If $(T',f')$ is obtained by a perturbation of type \textbf{P2} or \textbf{P3} of $(T,f)$, then
    \begin{equation*}
        d_1(\TMD(T,f),\TMD(T',f')) \leq \frac{5}{2}\epsilon.
    \end{equation*}
    Furthermore, if the vertex $v$ given by the perturbation of type \textbf{P2} or \textbf{P3} is such that $f'(v)\leq f(l)$ for some leaf $l$ of $T$ that is a descendant of $v$, then
    \begin{equation*}
        d_1(\TMD(T,f),\TMD(T',f')) \leq \epsilon.
    \end{equation*}
    
    \item If $(T',f')$ is obtained by a perturbation of type \textbf{P4} of $(T,f)$, then
    \begin{equation*}
        d_1(\TMD(T,f),\TMD(T',f')) \leq \frac{5}{2}\epsilon.
    \end{equation*}
    Furthermore, if the vertex $v$ in $T$ given by the perturbation of type \textbf{P2} or \textbf{P3} from $(T',f')$ to $(T,f)$ is such that $f(v)\leq f'(l)$ for some leaf $l$ of $T'$ that is a descendant of $v$, then
    \begin{equation*}
        d_1(\TMD(T,f),\TMD(T',f')) \leq \epsilon.
    \end{equation*}
\end{enumerate}
\label{thm:main}
\end{theorem}

Note that the 1-Wasserstein stability results are weaker than those obtained for the bottleneck distance in \cite[SI, Theorem 1]{kanari2018topological}. First, the Lipschitz constant we provide in statement 1 of the theorem is dependent on the number of leaves in $T$. Second, for statements 2 and 3, we are not allowing for arbitrarily intervals of a fixed length to be attached to or removed from $T$. With respect to the 1-Wasserstein distance, it is not hard to see that this kind of perturbation is unstable. For instance, if $T$ is a single root vertex equipped with the zero map $f$, and we obtain $(T',f')$ by attaching $n$ vertices to this root, each with $f'$ value $\epsilon$, we can make the 1-Wasserstein distance between the two associated TMDs arbitrarily large by increasing $n$. 

We will first prove statement 1 of this theorem, then statement 2. Statement 3 follows immediately from statement 2, since these statements are identical with the roles of $(T,f)$ and $(T',f')$ reversed. Our methods of proof are based on those of \cite[SI, Section 4]{kanari2018topological}. The following proposition is equivalent to statement 1.

\begin{proposition}
Let $T$ be a rooted tree with root $r$ and $f,f':N(T) \to \mathbb{R}$ such that $|f(v)-f'(v)|\leq \epsilon$ for all vertices $v$ in $N(T)$. Then
\begin{equation*}
    d_1(\TMDg(T,f),\TMDg(T,f')) \leq |L(T)|\epsilon.
\end{equation*}
\label{prop:fpurt}
\end{proposition}
\begin{proof}
If $T$ is a single vertex, then $T$ has one leaf, which is also its root, and the fact that
\begin{equation*}
    d_1(\TMDg(T,f),\TMDg(T,f'))\leq \epsilon
\end{equation*}
is clear.

We proceed with the remaining cases by induction on the depth of $T$, each time proving that every point not in $\diag$ in each persistence diagram can simultaneously be matched with another point not in $\diag$ that is $\epsilon$ away in sup norm, provided $T$ is of depth $d$. This will imply the desired result since the TMD algorithm produces exactly one interval for each leaf. We will assume inductively that the matching we get for trees of depth $d-1$ between functions $f$ and $f'$ sends a copy of $(f(r),L)$ to a point with first coordinate $f'(r)$, and whose inverse sends a copy of $(f'(r),L')$ to a point with first coordinate $f(r)$, where $L$ and $L'$ are the maximum values of $f$ and $f'$ on the leaves of $T$ respectively. This additional constraint can clearly be satisfied for the base case where $T$ has depth zero, and hence is a single vertex.

Suppose $T$ is a tree of depth $d$ with $f$ and $f'$ as in the statement of the theorem. Let $c_1,\ldots,c_n$ denote the children of the root of $T$. Each $c_i$ is the root of a subtree $T_i$ of $T$ whose vertices are $c_i$, its children, its children's children, and so on. Let $b_i$ be the leaf that takes the greatest value of $f$ subject to the constraint that $c_i$ is an ancestor of $b_i$. We define $b_i'$ similarly for $f'$. Thus,
\begin{equation*}
\begin{split}
    \TMDg(T,f) &= \diag \sqcup \{ (f(r), f(b_i)): 1 \leq i \leq n \} \sqcup \bigsqcup_{i = 1}^n D_i, \\
    \TMDg(T,f') &= \diag \sqcup \{ (f'(r),f'(b'_i)): 1 \leq i \leq n \} \sqcup \bigsqcup_{i = 1}^n D'_i,
\end{split}
\end{equation*}
where
\begin{equation*}
\begin{split}
    D_i &= \TMDg(T_i,f) - \{(f(c_i), f(b_i))\},\\
    D'_i &= \TMDg(T_i,f') - \{(f'(c_i), f'(b'_i))\}.
\end{split}
\end{equation*}
If we can show that there is a matching of $D_i$ and $D'_i$ that sends no point more than $\epsilon$ away with respect to the sup norm, then the matching that sends each $(f(r),f(b_i))$ to $(f'(r),f(b'_i))$ does the same for $\TMDg(T,f)$ and $\TMDg(T,f')$. Indeed, assuming without loss of generality that $f(b_i) \geq f'(b'_i)$,
\begin{equation*}
|f(b_i) - f'(b'_i)| = f(b_i) - \max_{l\in L(T_i)} f'(l) \leq f(b_i) - f'(b_i) \leq \epsilon   
\end{equation*}
Further, this matching and its inverse also satisfy our additional inductive hypothesis regarding points with left coordinate equal to $f(r)$ and $f'(r)$. Since there are exactly $|L(T)|$ points in both persistence diagrams excluding $\diag$ we will be done once we can find a desired matching between $D_i$ and $D'_i$.

By inductive hypothesis we already have such a matching between $\TMDg(T_i,f)$ and $\TMDg(T_i,f')$, and we will adjust this matching to produce a matching between $D_i$ and $D'_i$. If our given matching sends $(f(c_i),f(b_i))$ to $(f'(c_i), f'(b'_i))$ then we immediately get a desired matching between $D_i$ and $D'_i$.

Otherwise, let $\phi:\TMDg(T_i,f) \to \TMDg(T_i,f')$ be our matching and let $(x',y') = \phi((f(c_i), f(b_i)))$, $(x,y) = \phi^{-1}((f'(c_i), f'(b'_i)))$. Thus  we may define $\tilde\phi:D_i\to D'_i$ by sending $(x,y)$ to $(x',y')$ and every other point to its image via $\phi$. By our inductive hypothesis, $x = f(c_i)$ and $x' = f'(c_i)$. Thus $|x-x'| \leq \epsilon$. Assume without loss of generality that $f'(b'_i) \leq f(b_i)$. Now assume, \textit{with} loss of generality, that $f'(b'_i) \geq y$. Thus we have the equation.
\begin{equation}
    |y-y'| = |y-f'(b'_i)+f'(b'_i)-y'|.
    \label{eqn:diffs}
\end{equation}
Both differences on the right side of the above equation are bounded by $\epsilon$, the first due to the matching $\phi$, and the last since by definition of $b'_i$, $f'(b'_i) \geq y'$ and so
\begin{equation*}
    |f'(b'_i)-y'|  = f'(b'_i)-y' \leq f(b_i) - y' \leq \epsilon.
\end{equation*}
Notice also that the last difference in Equation \ref{eqn:diffs} is nonnegative. Meanwhile, $y-f'(b'_i)$ is nonpositive since $f'(b'_i) \geq y$. Thus, Equation \ref{eqn:diffs} shows that $|y-y'|$ is a difference of nonnegative numbers each bounded by $\epsilon$, and hence itself is less than or equal to $\epsilon$. Hence $\tilde\phi$ is a matching which alters each coordinate of every point by no more than $\epsilon$.

Finally, still assuming that $f'(b'_i) \leq f(b_i)$, we also assume $f'(b'_i) \leq y$. Now the relevant equation is
\begin{equation}
    |y-y'| = |y-f(b_i)+f(b_i)-y'|.
    \label{eqn:diffs2}
\end{equation}
Once again, both differences on the right are bounded by $\epsilon$, the second due to the matching $\phi$, and the first since by definition of $b_i$, $y\leq f(b_i)$, and so
\begin{equation*}
    |y - f(b_i)| = f(b_i) - y \leq f(b_i) - f'(b'_i) \leq \epsilon. 
\end{equation*}
As before, we have also just shown the first difference in the right side of Equation \ref{eqn:diffs2} is nonpositive, while the second is nonnegative since $f(b_i) \geq f'(b'_i) \geq y'$. Hence $|y-y'|$ is the difference of two nonnegative numbers bounded by $\epsilon$, and so is itself less than $\epsilon$. Again, this implies that $\tilde\phi$ is a matching which sends no point more than $\epsilon$ away with respect to the sup norm and the proof is complete.
\end{proof}

\begin{remark}
Kanari et al. \cite[SI]{kanari2018topological} show perturbations of type \textbf{P1} are bottleneck stable by constructing a matching that sends points not in $\diag$ to each other. Proposition \ref{prop:fpurt} follows from their matching and the fact that the TMD algorithm produces a diagram with $|L(T)|$ points excluding $\diag$. The construction of the matching in \cite{kanari2018topological} implicitly assumes that $b_i = b'_i$. Here, we use a different inductive hypothesis and therefore include our own full construction for completeness.
\end{remark}

In fact, for our above proposition to hold, the values of $f$ and $f'$ only need to be close on the root, branch points, and leaves- the following proposition shows we can remove all other vertices without changing TMDg.

\begin{proposition}
Let $T$ be a tree with root $r$ and an edge $(a,b)$. Let $T'$ be the tree given by $N(T') = N(T) \cup \{v'\}$ and $E(T') = E(T) \cup \{(a,v'),(v',b)\} - \{(a,b)\}$, with root $r$. Let $f:N(T) \to \mathbb{R}$ and $f':N(T') \to \mathbb{R}$ be such that $f'$ restricts to $f$ on $N(T)$. Then
\begin{equation*}
    \TMDg(T,f) = \TMDg(T',f').
\end{equation*}
\label{prop:intnodes}
\end{proposition}
\begin{proof}
We prove this by induction on the depth of $T$. If $T$ has depth zero, then $T$ is a single vertex and the theorem is vacuously true since $T$ has no edge.

Otherwise, assume the result holds for adding vertices to trees of depth less than $d$ and $T$ is of depth $d$. Then, first assume neither $a$ nor $b$ is the root of $T$. Thus, adopting the notation of $c_i$, $b_i$, $T_i$, and $D_i$ from the previous proof, we observe that $a$ and $b$ must in some $T_j$. We define $T'_j$ by $N(T'_j) = N(T_j) \cup \{v'\}$ and $E(T'_j) = E(T_j) \cup \{(a,v'),(v',b)\} - \{(a,b)\}$. We further define $D'_j = \TMDg(T'_j,f') - \{(f(r), f(b_j))\}$. It follows that
\begin{equation*}
    \begin{split}
        \TMDg(T,f) &= \diag \sqcup \{ (f(r), f(b_i)): 1 \leq i \leq n \} \sqcup \bigsqcup_{i = 1}^n D_i, \\
    \TMDg(T',f') &= \diag \sqcup \{ (f(r), f(b_i)): 1 \leq i \leq n \} \sqcup D'_j\sqcup\bigsqcup_{i \neq j} D_i.
    \end{split}
\end{equation*}
But by our inductive hypothesis these $D_j$ and $D'_j$ are identical, so the above persistence diagrams must be identical as well.
It remains to check the case where either $a$ or $b$ is the root. Without loss of generality, suppose $a$ is the root of $T$. Then $b=c_j$ for some $j$. Let $T'_j$ be given by $N(T'_j) = N(T_j) \cup \{v'\}$ and $E(T'_j) = E(T_j) \cup \{(v',c_j)\}$. Thus
\begin{equation*}
    \TMDg(T',f') = \diag \sqcup \{ (f(r), f(b_i)): 1 \leq i \leq n \} \sqcup D'_j\sqcup\bigsqcup_{i \neq j} D_i,
\end{equation*}
where
\begin{equation*}
    D'_j := \TMDg(T'_j,f') - \{(f'(v'), f(b_i))\}= \TMDg(T_j,f) - \{(f(c_j), f(b_i))\}= D_j.
\end{equation*}
Hence $\TMDg(T,f) = \TMDg(T',f')$.
\end{proof}

The following proposition shows that statement 2 of Theorem \ref{thm:main} holds for perturbations of type \textbf{P2}. By applying the previous proposition as well, the statement also holds for perturbations of type \textbf{P3}.

\begin{proposition}
Let $T$ be a rooted tree with root $r$ and $f:N(T) \to \mathbb{R}$. Let $v$ be any vertex of $T$ and $T'$ be another tree given by $N(T') = N(T) \cup \{v'\}$ and $E(T') = E(T)\cup \{(v,v')\}$, with root $r$. Let $f':N(T') \to \mathbb{R}$ be such that $f'=f$ on $N(T)$ and $|f'(v')-f'(v)| \leq \epsilon$. If $f(v) \leq f(l) + \delta$ for at least one leaf $l$ which is a descendant of $v$ then
\begin{equation*}
    d_1(\TMDg(T,f),\TMDg(T',f')) \leq \epsilon + \frac{3}{2}\delta.
\end{equation*}
\label{prop:addbranch}
\end{proposition}
\begin{proof}
To avoid certain mild irritations during the proof, we will first prove the proposition when $f'$ (and hence also $f$) is injective. We will then attain the general case by perturbing $f'$ via Proposition~\ref{prop:fpurt}.

For convenience, let $\eta = \epsilon + \frac{3}{2}\delta$. In this proof we proceed by induction on the depth $d$ of $v$ instead of the depth of $T$. As our inductive hypothesis, we assume that when $v$ has depth less than $d$ in a tree $T$, and $T'$ satisfies the hypothesis of the proposition the following is true. If $l$ maximizes $f$ on the leaves of $T$ and $l'$ maximizes $f'$ on the leaves of $T'$, then we can find a $\eta$-matching sending $(f(r), f'(l'))$ to $(f(r),f(l))$. Hence for our induction step we will need to construct a matching with this additional property.
 
For our base case $v = r$. As in the previous proofs, we may write
\begin{equation*}
    \TMDg(T,f) = \diag \sqcup \{ (f(r), f(b_i)): 1 \leq i \leq n \} \sqcup \bigsqcup_{i = 1}^n D_i
\end{equation*}
with $D_i$, $T_i$, $b_i$, and $c_i$ defined as in the proof of Proposition \ref{prop:fpurt}. If we attach the vertex $v'$ to $r$, then we have
\begin{equation*}
        \TMDg(T',f') = \diag\sqcup \{(f(r), f'(v'))\} \sqcup \{ (f(r), f(b_i)): 1 \leq i \leq n \} \sqcup \bigsqcup_{i = 1}^n D_i
\end{equation*}
Hence by sending $(f(r), f'(v'))$ to $(f(r),f(r))$ on the diagonal, and everything else to itself we obtain a $\eta$-matching between $\TMDg(T',f')$ and $\TMDg(T,f)$. However, if $v'$ maximizes $f'$ on the leaves of $T'$ and the depth of $T$ is nonzero then this matching does not satisfy the inductive hypothesis. Whenever this happens, we can instead choose the $\eta$-matching with the following mappings
\begin{equation*}
    \begin{split}
        (f(r), f'(v'))&\longmapsto (f(r), f(b_k))\\
        (f(r), f(b_k))&\longmapsto ((f(r)+f(b_k))/2,(f(r)+f(b_k))/2)\\
    \end{split}
\end{equation*} that sends everything else to itself, where $b_k$ is the leaf maximizing $f$ on the leaves of $T$. Notice that $|f'(v')-f(r)|\leq \epsilon$, and further $f(r)-f(b_k) = f(v)-f(b_k) \leq \delta$, so if $f(b_k) < f(r)$ then this is a $\eta$-matching. Otherwise, $f(r) < f(b_k) < f(v')$, and
\begin{equation*}
    \begin{split}
        |f'(v')-f(b_k)| + |f(b_k) - (f(r)+f(b_k))/2| &\leq  |f'(v')-f(b_k)| + |f(b_k) - f(r)|\\
        &= f'(v')-f(b_k) + f(b_k) - f(r)\\
        &= f'(v') - f(r) \\ 
        &\leq \epsilon,
    \end{split}
\end{equation*}
so the proposed matching is an $\epsilon$-matching, and hence a $\eta$-matching.

Proceeding now by induction, we now assume the inductive hypothesis holds whenever the $v$ has depth less than $d>0$. Assume now that the depth of $v$ is $d$. Hence $v \neq r$, so then $v$ lies in $T_j$ for some $j$. We define $T'_j$ to be the tree with $N(T'_j) = N(T_j) \cup \{v'\}$ and  $E(T'_j) = E(T_j)\cup (v,v')$, and $b'_j$ to be the leaf in $T'_j$ that maximizes $f'$ over the leaves of $T'_j$. Thus
\begin{equation*}
        \TMDg(T',f') = \diag\sqcup \{ (f(r),f(b_i)): i\neq j \}\sqcup \{(f(r), f'(b'_j))\} \sqcup D'_j \sqcup  \bigsqcup_{i \neq j} D_i,
\end{equation*}
where
\begin{equation*}
    D'_j = \TMDg(T'_j,f') - \{(f(c_i), f'(b'_i))\}
\end{equation*}
For our convenience, we let $\gamma$ denote the number $\eta - f'(b'_j) + f(b_j)$. Notice that $v$ has depth $d-1$ in $T_j$. Consider the $\eta$-matching $\phi$ between $\TMDg(T'_j,f')$ and $\TMDg(T_j,f)$ guaranteed to exist by our inductive hypothesis. Since our inductive hypothesis also demands that $\phi((f(c_j), f(b'_j))) = (f(c_j), f(b_j))$, $\phi$ restricts to a $\gamma$-matching between $D'_j$ and $D_j$.

We can use this matching to construct a $\eta$-matching between $\TMDg(T',f')$ to $\TMDg(T,f)$ by sending $(f(r), f'(b'_j))$ to $(f(r), f(b_j))$, $D'_j$ to $D_j$ via $\phi$, and everything else to itself. However, this matching may not satisfy our additional inductive assumption that points with right coordinates corresponding to leaves maximizing $f$ are sent to each other. If this is the case, then either $b'_j$ is the leaf maximizing $f'$ or $b_j$ is the leaf maximizing $f$, but not both. Since $f'(b'_j) \geq f(b_j)$, it must be the case that $b'_j$ maximizes $f'$, but $b_j$ does not maximize $f$. Let $b_k$ be the leaf that maximizes $f$. In this case, consider the matching $\psi$:
\begin{equation*}
    \begin{split}
        (f(r), f(b'_j))&\longmapsto (f(r), f(b_k)) \\
        (f(r), f(b_k))&\longmapsto (f(r), f(b_j)) \\
        D'_j &\longmapsto D_j,
    \end{split}
\end{equation*}
that sends everything else to itself. Of course, we define the restriction of $\psi$ to $D'_j$ to be $\phi$. Since $f'(b'_j) \geq f(b_k) \geq f(b_j)$, we observe
\begin{equation*}
    \gamma + |f'(b'_j) - f(b_k)| + |f(b_k) - f(b_j)| = \gamma + f'(b'_j) - f(b_j) = \eta.
\end{equation*}
As $\phi$ is a $\gamma$-matching, the fact that $\psi$ is a $\eta$-matching satisfying our inductive hypothesis is immediate. 

Now suppose $f'$ on $T'$ is not injective. In this case, we can always find an injective $g'$ on $N(T')$ satisfying both that $\|f'-g'\|_\infty \leq \kappa$ for any fixed positive real $\kappa$ and the hypotheses of the proposition with $\epsilon$ replaced by $\epsilon + 2\kappa$ and $\delta$ replaced by $\delta + 2\kappa$. In fact, this will always be the case whenever we perturb $f'$ at each vertex by a value no more than $\kappa$ in such a way that the perturbed function is injective. We let $g$ be the restriction of $g'$ to $N(T)$. Hence by Proposition \ref{prop:fpurt} and the fact that $g'$ is injective,
\begin{equation*}
\begin{split}
    d_1(\TMDg(T,f),\TMDg(T',f')) &\leq d_1(\TMDg(T,f),\TMDg(T,g))\\
    &+ d_1(\TMDg(T,g),\TMDg(T',g'))\\
    &+ d_1(\TMDg(T',g'),\TMDg(T',f')),
\end{split}
\end{equation*}
and this is less than or equal to
\begin{equation*}
    (|L(T)| + |L(T')|)\kappa + 5\kappa + \eta .
\end{equation*}
Considering a sequence of values for $\kappa$ approaching zero the result is immediate.
\end{proof}

The bound given by this proposition is sharp. We show this with an example in Figure \ref{fig:sharpbound}. Indeed the tree on the right of this figure is produced by a perturbation of the tree on the left of the same type as in the hypothesis of the proposition. The tree on the left has TMDg equal to $(0,K-\delta) \sqcup \diag$ while the tree on the right has TMDg equal to $(0,K+\epsilon)\sqcup(K, K-\delta)\sqcup \diag$. It is not difficult to see that for sufficiently large $K$, the optimal matching from the later diagram to the former is the matching which sends $(0, K+\epsilon)$ to $(0, K-\delta)$, $(K, K-\delta)$ to the point $(K-\delta/2,K-\delta/2)$ on the diagonal, and everything else to itself.

\begin{figure}[htbp]
\centering
\resizebox{.5\textwidth}{!}{
\includegraphics{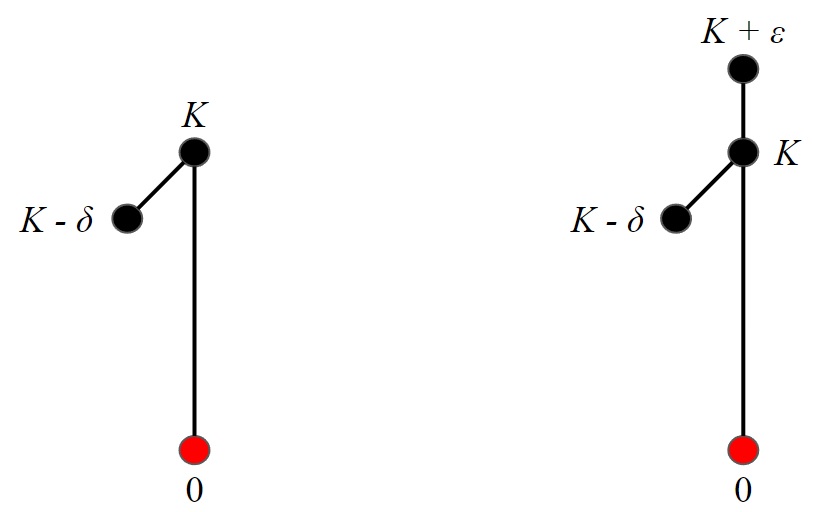}
}

    \caption{A pair of trees with roots in red showing that the bound in Proposition \ref{prop:addbranch} is sharp.}
    \label{fig:sharpbound}
\end{figure}

\section{Acknowledgements}
AG is grateful for the support by the Engineering and Physical Sciences Research Council of Great Britain under research grants EP/R020205/1. HAH gratefully acknowledges EPSRC EP/R005125/1 and EP/T001968/1, the Royal Society RGF$\backslash$EA$\backslash$201074 and UF150238. DB and HAH are members of the Centre for Topological Data Analysis, funded in part by EPSRC EP/R018472/1, and grateful for helpful discussions with centre members. For the purpose of Open Access, the authors have applied a CC BY public copyright licence to any Author Accepted Manuscript (AAM) version arising from this submission.

\bibliographystyle{abbrv}
\bibliography{refs}

\end{document}